\documentclass[11pt]{book}
\usepackage{Wiley-AuthoringTemplate}
\usepackage{chapterbib}
\usepackage{subcaption}
\usepackage[justification=centering]{caption}

\usepackage{tikz}
\usetikzlibrary{arrows,automata}
\usepackage[utf8]{inputenc}\usepackage{amsmath,amssymb,amsfonts}

\usetikzlibrary{shapes.geometric,fit}
\tikzstyle{container} = [draw, rectangle, inner sep=1.5cm]

  \tikzset{main node/.style={circle,draw,minimum size=1cm,inner sep=0pt},}
\usepackage[sectionbib,authoryear]{natbib}
\usepackage{acronym}
\usepackage{todonotes}
\newif\ifuseboldmathops
\newif\ifuseittextabbrevs
\useboldmathopstrue   

\ifuseittextabbrevs

	\newcommand{\ie}{{\it i.e.}}

\else

	\newcommand{\ie}{i.e.}

\fi

\ifuseboldmathops

\else

\fi

\ifuseboldmathops

\else

\fi

\ifuseboldmathops

\else

\fi

\ifuseboldmathops

\else

\fi

\newcommand{\Always}{\Box \, }
\newcommand{\Eventually}{\Diamond \, }

\newcommand{\until}{\mbox{$\, {\sf U}\,$}}

\newcommand{\calAP}{\mathcal{AP}}

\newcommand{\win}{\mathsf{Win}} 

\acrodef{mdp}[MDP]{Markov Decision Process}
\acrodef{pomdp}[POMDP]{Partially Observable Markov Decision Process}

\newcommand{\hgame}{\mathcal{HG}}

\newcommand{\calA}{\mathcal{A}}
\newcommand{\game}{\mathcal{G}}

\acrodef{dfa}[DFA]{Deterministic Finite-State Automaton}
\acrodef{scltl}[scLTL]{syntactically co-safe LTL}
\acrodef{ltl}[LTL]{Linear Temporal Logic}

\acrodef{vm}[VM]{Virtual Machine}

\newcommand{\servs}{\mathsf{Servs}}

\makeindex

\begin{document}


\chapter[]{
A Theory of Hypergames on Graphs for Synthesizing Dynamic Cyber Defense with Deception}

\author[1]{Abhishek N. Kulkarni}
\author*[2]{Jie Fu}

\address[1]{\orgdiv{Robotics Engineering Program}, 
\orgname{Worcester Polytechnic Institute}, 

     \street{100 Institute Rd.}, \city{Worcester, MA},  \country{US}.}%

\address[2]{\orgdiv{Dept. of Electrical and Computer Engineering, Robotics Engineering Program}, 
\orgname{Worcester Polytechnic Institute}, 
      \street{100 Institute Rd.}, \city{Worcester, MA},\country{US}}%

\address*{Corresponding Author: Jie Fu; \email{jfu2@wpi.edu}}

\maketitle
 
\begin{abstract}{Abstract}
In this chapter, we present an approach using formal methods to   synthesize reactive defense strategy in a cyber network, equipped with a set of decoy systems. We first generalize formal graphical security models--attack graphs--to incorporate defender's countermeasures in a game-theoretic model, called an attack-defend game on graph. This game captures  the dynamic interactions between the defender and the attacker and their defense/attack objectives in formal logic. Then, we  introduce a class of hypergames to model asymmetric information created by decoys in the attacker-defender interactions. Given qualitative security specifications in formal logic, we show that the solution concepts from hypergames and reactive synthesis in formal methods can be extended to synthesize effective dynamic defense strategy using cyber deception. The strategy takes the advantages of the misperception of the attacker to ensure security specification is satisfied, which may not be satisfiable when the information is symmetric.
 \end{abstract}


\section{Introduction}

Cyber deception is a key technique in  network defense. With cyber deception, the defender creates uncertainties and unknowns for the attacker. By doing so, the attacker’s strategy in exploiting the system becomes less effective, thus, resulting in improved security and safety of the network. In this chapter, we investigate a formal methods approach for synthesizing defensive strategies in cyber network systems with cyber deception. We employ formal security specifications to  express a rich class of desired properties. For example, a defender may need to satisfy a safety property in terms of preventing the attacker from reaching critical data server. He may also need to satisfy a liveness property stating that a service should eventually  be provided to the user after being made temporarily unavailable.
Given formal  security specifications,  formal synthesis is to compute a defense strategy, if exists, with which the defender can provably satisfy his specification against all possible actions from the attacker.

Formal methods have been employed to verify the security of network systems.  Formal graphical security models such as attack graphs \citep{Jha2002Two} and attack trees \citep{schneierbruceAttackTrees2007} are used in model-based verification of system security.  An attack graph captures multiple paths that an attacker can carry out by exploiting vulnerabilities and their dependencies in a network to reach the attack goal. Given an attack graph, the formal security specification can be verified using model checking algorithms for transition systems \citep{baierPrinciplesModelChecking2008}. 
An attack tree builds a tree structure that describes how the attacker can achieve his goal by achieving a  set of subgoals. The root of the tree is the main attack goal and the leaves of the tree are elementary attack subgoals. The internal tree nodes shows the logical dependency between subgoals at different levels of the tree.
To incorporate defender's counter-measures, attack-defense trees \citep{kordy2010foundations,kordyDAGbasedAttackDefense2014} are proposed to capture the dependencies between actions and subgoals for both attacker and defender. These models are used in verifying quantitative security properties in temporal logic \citep{aslanyanQuantitativeVerificationSynthesis2016a,hansen2017quantitative,kordyQuantitativeAnalysisAttack2018}. The major limitation of attack trees is that it does not characterize network status changes under the attack actions and thus  may fail to generate some attack scenarios. It is also noted that these formal graphical models do not capture the asymmetric information between the attacker and the defender due to cyber deception. Specifically, these models assume both defender and attacker knows the game they are playing, while as with cyberdeception, the defender intentionally introduces incorrect or uncertain information about the game to the attacker.


Active deception \citep{Jajodia2016Cyber} employs decoy systems and other defenses, including access control and online network reconfiguration, to conduct deceptive planning against the intrusion of malicious users who have been detected and confirmed by sensing systems. To design defense strategies with deception,
game theory has been employed \citep{Hor2012Manipulating,Huang0Dynamic, Horak2019Compact,Cohen2006Use,Zhu2018On}. These game-theoretic models express the attacker and defender’s objectives using reward/loss functions. In \citep{Horak2019Compact}, a partially observable stochastic game is formulated to capture the interaction between an attacker and a defender with one-sided partial observations. The attacker is to exploit and compromise the system without being detected and has complete observation. The defender is to detect the attacker and reconfigure the honeypots.  \citet{Hor2012Manipulating} consider the case when the attacker has incomplete information and forms a belief about the defender’s unit. Players employ Bayesian rules to update the belief about the state in the game. Leveraging the attacker’s incomplete information, the defender may mislead the attacker’s belief and thus his actions to minimize the damage to the network measured by a state-dependent loss function. However, reward and loss functions are not expressive enough to capture more complex qualitative defense/attack objectives studied in attack graph, such as safety and temporally extended attack goals. These objectives can be captured succinctly using temporal logic \citep{mannaTemporalLogicReactive1992}. When formal specification is used in specifying defense objectives, there is a lack of formal synthesis methods which employ cyber deception to  ensure the security goals are met.

We study the problem of formal synthesis of secured network systems with active cyber deception. We view the interactions between the defender and the attacker as a two-player game played on a finite graph. Combining the game graph abstraction with the logical security specifications, we construct a model of an attack-defend game as a game on a graph with temporal logic objectives \citep{Pnueli1989On,Chatterjee2012survey}. This game includes both the controllable and uncontrollable actions to represent the actions and exploits by the defender and the attacker, respectively. 

In such a game between a defender and an attacker,  the attacker plays with incomplete information, if he does not know the locations of honeypots. Furthermore, if the attacker mistakes a honeypot as a critical host, then we say that he has a misperception about the game. We extend the theory of hypergame to reason about the asymmetric incomplete information between players and to enable synthesis of deceptive strategies. A hypergame \citep{Bennett1980Hypergames,Vane2000Using,Kovach2015Hypergame} is a game of perceptual games, \ie, games perceived by individual players given the information available to them and the higher-order information known to them, \ie, what the player knows about the information known to the opponent. 
Based on the hypergame modeling, the key questions are: how will the attacker carry out his attack mission, given his incomplete or incorrect information? And, how to synthesize effective defense strategies, which leverage the defender’s private information to ensure that the defender's logical security specifications are satisfied?

Our insight is that deception with honeypots can create a misperception about the labeling function of the attack-defend game. A labeling function relates an outcome—a sequence of states in the game graph—to the properties specified in logic. When honeypots are introduced, an attacker might mislabel a honeypot as a critical host and pursue to reach it. Under this formulation, our main algorithmic contribution is the solution of hypergames under labeling misperception and linear temporal logic objectives. Our solution approach includes two steps: The first step is to synthesize the rational attack strategy perceived by the attacker using solutions of omega-regular games \citep{Chatterjee2012survey,Zielonka1998Infinite}. The synthesized strategy serves as a predictive model of the attacker’s rational behavior, which is then used to refine the original game graph to eliminate actions perceived to be irrational from the attacker’s perspective. In the second step, a level-2 hypergame is solved, yielding a deceptive defense strategy, if one exists, that ensures a specification is  satisfied with probability one, given the misperception of the attacker.  A case study is employed to illustrate how to apply the game-theoretic reasoning to synthesize  deceptive strategies.

We structure the remainder of the chapter so as to provide rigorous mathematical treatment of the topic for a reader familiar with formal methods, and support it with elaborate descriptions, discussions and examples to illustrate our approach to a reader new to the area.

\section{Attack-Defend Games on Graph}



In this section, we introduce a model, called \textit{An Attack-Defend (AD) Game on a Graph}, that augments the attack graph model with the defense actions available to the defender. Our AD game on graph model resembles the game on graph model, which is commonly used in reactive synthesis \citep{Pnueli1989On}. 


Formally, an AD game on graph can be written as a tuple $\game = \langle G, \varphi \rangle$ where the two main components are (i) $G$: a game arena, and (ii) $\varphi$: the Boolean payoff function (\ac{ltl} specification) of the defender. Let us understand each of the component in more detail.



\subsection{Game Arena}
A game arena is a transition system with labels assigned to the states. It captures different configurations of the network and the actions that the attacker and the defender may use to change the current configuration. A configuration of system is a set of state variables that jointly define the current state of the system. For instance, a state variable may be a collection of the current host compromised by the attacker, IP addresses of different hosts over the network, an enumeration of services running over each host, or a list of users currently accessing the hosts with their privileges (root, user, none).  Suppose that there are $n$ state variables and we denote the $i$-th state variable as $X_i$, then the domain of a state-space can be given by $S = X_1 \times X_2 \times \ldots \times X_n$. 
Given this notion of state, we formally define a game arena as follows:



\begin{definition}[Arena]
\label{def:arena}
A turn-based, deterministic game arena between two players P1 (defender, pronoun ``he") and P2 (attacker, pronoun ``she") is a tuple 
\[ G = \langle S,A, T , \calAP,  L \rangle,\] whose components are defined as follows:
\begin{itemize}
\item $S = S_1 \cup S_2$ is a finite set of  states partitioned into two sets $S_1$ and $S_2$. At a state  in $S_1$, P1 chooses an action and at a state  in $S_2$, P2 selects an action.

\item $A = A_1 \cup A_2$ is the set of actions. $A_1$ (resp., $A_2$) is the set of actions for P1 (resp., P2); 
\item $T : (S_1 \times A_1)\cup (S_2 \times A_2) \rightarrow S$ is a \emph{deterministic} transition function that maps  a state-action pair to a next state.
\item $\calAP$ is the set of atomic propositions.
\item $L: S\rightarrow 2^\calAP$ is the labeling function that maps
  each state $s\in S$ to a set $L(s)\subseteq \calAP$ of atomic
  propositions that evaluate to true at that state.
 \end{itemize}
\end{definition}
The last two components of the game arena are related to the security specifications.

We discuss an  example   to illustrate the above concept. 

\begin{example}
\label{ex:4-hosts}
\begin{figure}
    \centering
    \includegraphics[width=0.4\textwidth]{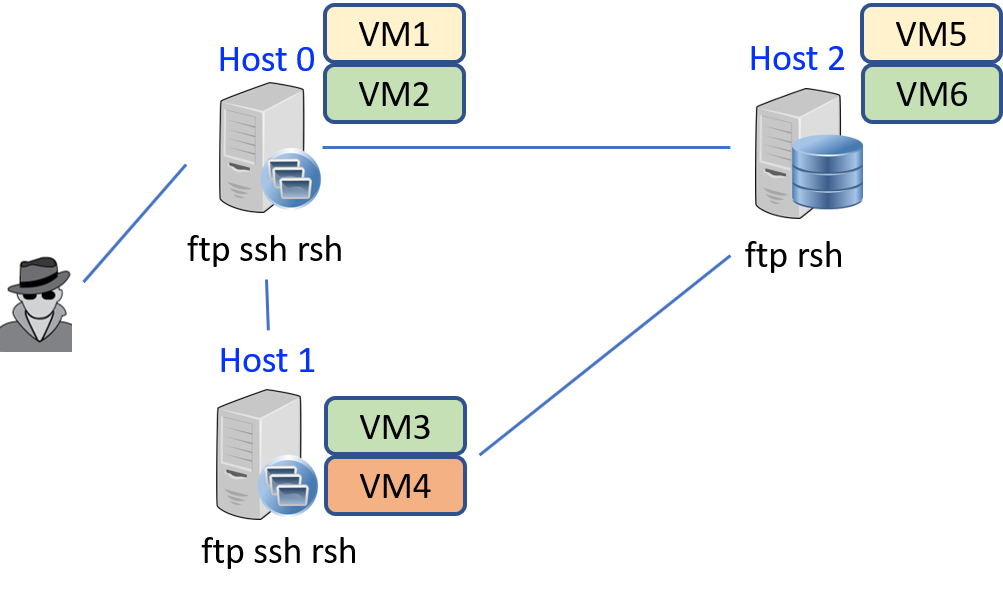}
    \caption{Configuration of the  system.}
    \label{fig:ex_network}
\end{figure}

Consider the system as shown in Fig.~\ref{fig:ex_network} consisting of three hosts with platform diversity. Each host can hold up to two \ac{vm}s with different operation systems and services. For a fixed configuration of \ac{vm}s, a fragment of the attack graph can be generated based on the set of known vulnerabilities, as shown in Fig.~\ref{fig:ex_attackgraph}. 

\begin{figure}[ht]
    \centering 
        \begin{subfigure}[t]{ \textwidth}
        \includegraphics[width=\textwidth]{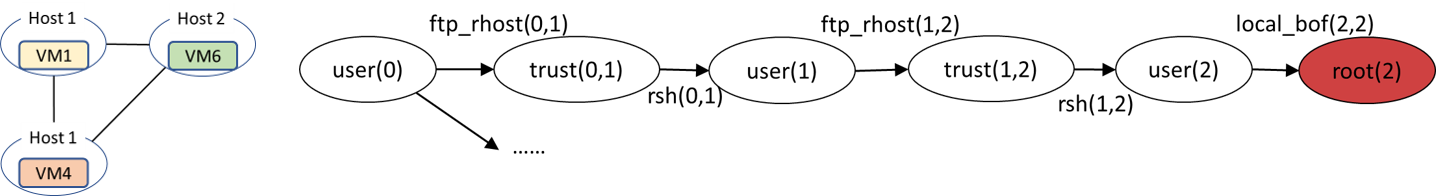}
            \caption{A fragment of the attack graph corresponding to one fixed configuration.
     \label{fig:ex_attackgraph}}
\end{subfigure}
        \begin{subfigure}[t]{0.6\textwidth}
        \includegraphics[width=\textwidth]{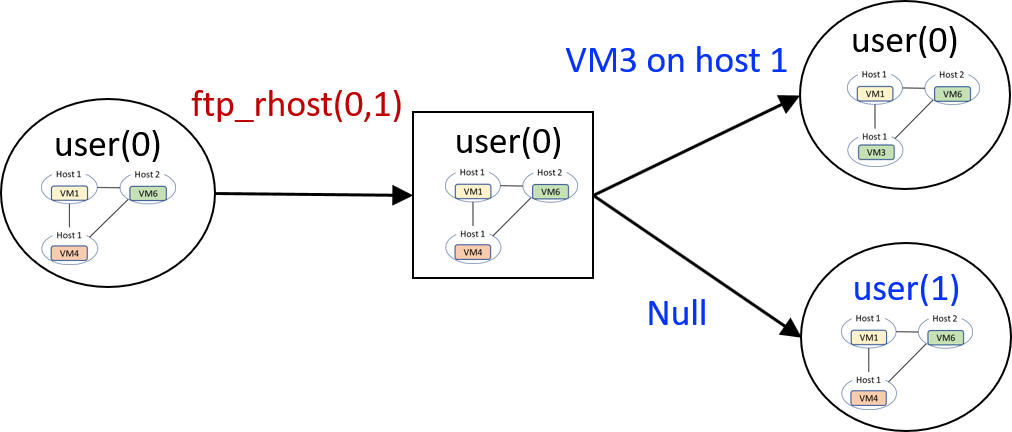}
            \caption{(A fragment of) the attack-defend game arena. 
     \label{fig:ex-game}}
\end{subfigure} 
\caption{The comparison between the attack graph and the attack-defend game arena.}
\label{fig:compare}
 \end{figure}
  In reactive defense, the defender has detected the attacker is in host 0 and can exploit the vulnerability to gain user access to host 1. In that case, the defender can change the platform in host 1 to be VM3, for which the attack action will not be effective. The interaction is then captured in the game on graph, shown in Fig.~\ref{fig:compare}. The action of the defender can also be stopping or running  a service at the managed endpoints, which are omitted. We distinguish the set of states into square states at which the defender makes a move and circle states at which the attacker makes a move. 
 
\end{example}


\subsection{Specifying the security properties in \ac{ltl}}
We consider qualitative formal specifications for defender and attacker objectives. Different from quantitative utility functions in terms of costs, qualitative logic formulas capture hard security constraints that the network defense system must satisfy.


The defender has two types of goals, namely (i) operational objectives; such as the services should eventually be available to the legitimate users, and (ii) defense objectives; such as the attacker should never be able to compromise servers with sensitive information. However, the intention of attacker is often unknown. Thus, we consider the worst-case scenario where the attacker's objective is to violate the security goal of the defender. 


We choose to express the security goal of the defender using \ac{ltl} \citep{mannaTemporalLogicReactive1992}. \ac{ltl} allows us to express the security properties of system with respect to time. We shall now present the formal syntax and semantics of \ac{ltl} and then discuss several examples.

Let $\calAP$ be a set of atomic propositions. Linear Temporal Logic
(\ac{ltl}) has the following syntax,
\[ \varphi := \top \mid \bot \mid p \mid \varphi \mid \neg\varphi \mid \varphi_1 \land \varphi_2 \mid \bigcirc \varphi \mid \varphi_1 {\until} \varphi_2, \] where 
\begin{itemize}
    \item $\top,\bot$ represent universally true and false, respectively.
    \item $p \in \calAP$ is an atomic proposition.
    \item $\bigcirc$  is a temporal operator called the ``next'' operator (see semantics below).
    \item $\until$ is a temporal operator called the ``until'' operator (see semantics below).
    \end{itemize}
    

Let $\Sigma\coloneqq 2^\calAP$ be the finite alphabet. 
Given a word $w\in \Sigma^\omega$, let $w[i]$ be the $i$-th element in the word and $w[i\ldots]$ be the subsequence of $w$ starting from the $i$-th element. For example, $w=abc$, $w[0]=a$ and $w[1\ldots] = bc$. Formally,
we have the following definition of the semantics: 
\begin{itemize}
    \item $w\models p$ if $p \in w[0]$;
    \item $w\models \neg p $ if $p \notin w[0]$;
    \item $w\models \varphi_1\land \varphi_2$ if $w \models \varphi_1$ and $w\models \varphi_2$.
    \item $w\models \bigcirc \varphi$ if $w[1\ldots] \models \varphi$.
    \item $w\models \varphi \until \psi$ if $\exists i \ge 0$, $w[i\ldots] \models \psi$ and $\forall 0\le j<i$, $w[j\ldots]\models \varphi$.  
\end{itemize}
From these two temporal operators ($\bigcirc, \until$), we define two
additional temporal operators: $\Eventually$ ``eventually'' and $\Always$
``always''. Formally, $\Eventually \varphi = \top \until
\varphi$ and $\Always \varphi = \neg \Eventually \neg \varphi$. For details about the syntax and semantics of \ac{ltl}, the readers are referred to \citep{mannaTemporalLogicReactive1992}.

Here we present some examples. Suppose $root(2)$ is an atomic proposition that the attacker has root privilege on host 2, then a safety property that attacker never has root privilege on host 2 can be written in \ac{ltl} as a formula $\varphi_1 = \Always \neg root(2)$, which is read as ``proposition $root(2)$ is \textit{always} false." Similarly, a property that the attacker first gains a user privilege on host 1  and then a root privilege on host 2 can be expressed using an \ac{ltl} formula $\varphi_2 = \Eventually(user(1) \land \Eventually root(2))$, which is read as ``\textit{eventually} proposition $user(1)$  becomes true and then the proposition $root(2)$ becomes true." In general, it is also possible to express properties such as recurrence (some event occurs infinitely often) or persistence (some property eventually becomes true, and remains true thereafter) using \ac{ltl}. However, in this chapter, we restrict ourselves to a sub-class of \ac{ltl} called \ac{scltl} \citep{kupferman2001model}. Using \ac{scltl} we can reason about the reachability and safety\footnote{Safety and reachability are dual problems. Hence, reasoning about the safety objectives can be done by reasoning about the dual reachability problem.} properties.

This concludes a brief introduction to the concept of AD games on graphs; which do not model the asymmetric incomplete information available with the players. In the next section, we extend the notion of hypergames that incorporate the different perceptions that players may have due to incomplete information available to them.


\section{Hypergames on Graphs}

A hypergame models the situation where different players perceive their interaction with other players differently, and consequently play different games in their own minds depending on their perception. We consider the case where the difference in perception arises because of incomplete and potentially incorrect information. For instance, suppose a subset of nodes in the network are honeypots, the attacker may mistake these to be true hosts. We formulate a hypergame to model the interaction between the defender and the attacker given asymmetric information.
 
First, let's review the definition of hypergames.
\begin{definition}[Hypergame \citep{Bennett1980Hypergames, Vane2000Using}]
Given two players, a game perceived by player $1$ is denoted by $\game_1$, and a game perceived by player $2$ is denoted by $\game_2$. A level-1 hypergame is defined as a tuple 
\[ 
    \hgame^1 = \langle \game_1, \game_2 \rangle, 
\] 
In a level-1 hypergame, none of the player's is aware of other player's perception. 

When one player becomes aware of the other player's (mis)perception, the interaction is captured by a level-2 two-player hypergame, defined as a tuple,
\[
    \hgame^2= \langle \hgame^1, \game_2 \rangle.
\] 
where P1 perceives the interaction as a level-1 hypergame (as P1 is aware of P2's game $\game_2$ in addition to his own) and P2 perceives the interaction as the game $\game_2$. 
\end{definition}

We refer to the games $\game_1$ (resp., $\game_2$) as P1's (resp., P2's) perceptual game in level-1 hypergame, and $\hgame^1$ as P1's perceptual game in level-2 hypergame. As P2 is not aware that she might be misperceiving the game, her perceptual game in level-2 hypergame is still $\game_2$.

In general, if P1 computes his strategy
by solving an $(m-1)$-th level hypergame and P2 computes her strategy using
an $n$-th level hypergame with  $n < m$, then the resulting hypergame is said
to be a level-$m$ hypergame given as
\[
    \hgame^m = \langle \hgame^{m-1}_1, \hgame^n_2 \rangle. 
\]
 
Next, we show that by introducing honeypots, the attacker's perceptual game deviates from the actual game. This mismatch occurs in the labeling function. 
Recall that a labeling function $L$ assigns every state in the game arena with a subset of atomic propositions that are true at that state. 
Let us consider a   network with decoys where attacker is not aware of which hosts are decoys. Suppose $p$ is a proposition that a host $h$ is a decoy. Then, defender's labeling function, say $L_1$, labels $h$ correctly as a decoy. However, the attacker's labeling function, say $L_2$, will incorrectly label $h$ as a regular host. Given a path $\rho \in S^\ast$ in the game arena, this path may satisfy the security specification as $L_1(\rho)\models \varphi$, in which case the defender obtains payoff $1$ and the attacker obtains payoff $0$. However, due to misperception in the labeling, the attacker may have $L_2(\rho) \models \neg \varphi$ and thus have a misperception of the payoff of the path. We capture this misperception and asymmetric information using the new class of hypergames, defined as follows:



\begin{definition}[A Hypergame on a Graph with One-sided Misperception of Labeling Function] 
\label{def:one-sided-misperception}
Let $G_1 = \langle S, A, T, \calAP, L_1 \rangle$ be the game arena as constructed by P1. Similarly, let $G_2 = \langle S, A, T, \calAP, L_2 \rangle$ be a game arena as constructed by P2 based on her perception. Let $\varphi$ be the defense objective of P1. Then, we construct two games $\game_1 = \langle G_1, \varphi \rangle$ and $\game_2 = \langle G_2, \varphi \rangle$. When P1 is aware of P2's misperception, \ie~ P1 knows $L_2$ and, therefore, $\game_2$, we have the model of their interaction as a hypergame of level-2,
\[
    \hgame^2 = \langle \hgame^1, \game_2 \rangle,
\]
where $\hgame^1 = \langle \game_1, \game_2 \rangle$ is a hypergame of level-1 and is P1's perceptual game. P2's perceptual game is $\game_2$. We say $\hgame^2$ to be a hypergame  on  a graph with one-sided misperception when the labeling function of P1 coincides with the ground-truth labeling function, \ie~ $L_1 = L$. 

\end{definition}


\section{Synthesis of Provably-Secure Defense Strategies using Hypergames on Graphs}

Given the hypergame model, we present a solution approach to automatically synthesize a strategy for defender such that, for every possible action of attacker, the strategy ensures that the security goals  (\ie, $\varphi$) of defender are satisfied. In order to understand the synthesis approach for hypergame, we first look at the conventional solution approach used for a game on graph \citep{mcnaughton1993infinite,Zielonka1998Infinite}.

\subsection{Synthesis of Reactive Defense Strategies} 

Recall that in an AD game on graph model; $\game = \langle G, \varphi \rangle$, we assume that the information available to both players is complete and symmetric. Under this assumption, the solution for game on graph can be computed by constructing a game transition system and then using an algorithm to identify the winning regions for the attacker and the defender.

Before we introduce the game transition system, let us visit the equivalence of an \ac{scltl} specification with a \ac{dfa}. 

\begin{definition}[Specification \ac{dfa}]
A \ac{dfa} is a tuple,
\[
    \calA = \langle Q, \Sigma, \delta, I, F \rangle, 
\]
where 
\begin{itemize}
    \item $Q$ is a finite set of \ac{dfa} states.
    \item $\Sigma = 2^\calAP$ is an alphabet.
    \item $\delta: Q \times \Sigma \rightarrow Q$ is a deterministic transition function. The transition function can be extended recursively as: $\delta(q, uv) = \delta(\delta(q,u),v)$ for some $u,v\in \Sigma^\ast$.
    \item $I \in Q$ is a unique initial state.
    \item $F \subseteq Q$ is a set of final states.
\end{itemize}
A word $w=  \sigma_0\sigma_1\ldots \sigma_n$ is accepted by the \ac{dfa} if and only if $\delta(q_0,w)\in F$. 
Given an \ac{scltl} specification $\varphi$, a \ac{dfa} $\calA$ is called a specification \ac{dfa} when  every word $w$ defined over the alphabet $\Sigma$ that satisfies $w \models \varphi$ is accepted by \ac{dfa} $\calA$.
\end{definition}

Using this notion of equivalence between a \ac{scltl} formula and \ac{dfa}, we define the game transition system as follows:

\begin{definition}[Game Transition System]
Let $\calA = \langle Q, \Sigma, \delta, I, F \rangle$ be a \ac{dfa} equivalent to the specification $\varphi$. Then, given $\game = \langle G, \varphi \rangle$, the game transition system, represented as $G \otimes \calA$, is the following tuple: 
\[
G \otimes \calA = \langle S \times Q, A, \Delta, (s_0, q_0), S \times F \rangle,
\]
where 
\begin{itemize}
    \item $S\times Q$ is a set of states partitioned into P1's states $S_1 \times Q$ and P2's states $S_2 \times Q$;
    \item $A = A_1 \cup A_2$ is the same set of actions as labeled transition system $G$;
    \item $\Delta: (S_1 \times Q \times A_1) \cup (S_2 \times Q \times A_2) \rightarrow S \times Q$ is a \textit{deterministic} transition function that maps a game  state  $(s,q) $ and an action $a$ to a next state $(s',q')$ where $s' = T(s,a,s')$ and $q' = \delta(q, L(s'))$.  
    \item $(s_0, q_0) \in S \times Q$ where $q_0 = \delta(I, L(s_0))$ is an initial state of the game transition system; and 
    \item $S \times F \subseteq S \times Q$ is a set of final states. 
\end{itemize}

\end{definition}




The following theorem is a well-known result in game theory \citep{mcnaughton1993infinite,Zielonka1998Infinite}.

\begin{theorem}[Determinacy of Game on Graph] \label{thm:determinancy} 
All two-player zero-sum deterministic turn-based games on graph are determined. 
\end{theorem}

Thm.~\ref{thm:determinancy} is a very important result because it provides us with a characterization of the game state-space. It states that, \textit{at any state in the game transition system, either the defender or the attacker has a winning strategy.} In other words, the state space of game transition system is divided into two sets, one consisting of states from which the defender is guaranteed to satisfy his security objectives, and the second consisting of states from which attacker has a strategy to violate the defender's objectives.

\begin{example}
We illustrate the game on graph using a toy example. Consider a network system where the defender can switch between two network topologies, giving rise to two attack graph under two network topologies (Shown in Fig.~\ref{fig:attackgraphAB}). For simplicity, as the graph is deterministic, we omit the attack action labels on the graph. 
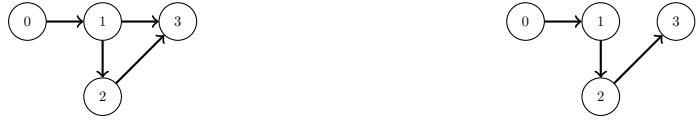
\begin{figure}[ht]
\begin{subfigure}[b]{0.5\textwidth}
    \centering
    \begin{tikzpicture}[align=center,  node distance=2cm, scale=0.5, transform shape]
    \node[main node] (0) {$0$};
    \node[main node] (1) [right of =0]  {$1$};
    \node[main node] (2) [below of =1] {$2$};
    \node[main node] (3) [right of =1] {$3$};
    
    \path[->,draw,thick]
    
    (0) edge node {} (1)
    (1) edge node {} (2)
    (1) edge node {} (3)
    (2) edge node {} (3);
\end{tikzpicture}
    \caption{Attack graph under topology A.}
    \label{fig:attackgraphA}
    \end{subfigure}
    \begin{subfigure}[b]{0.5\textwidth}
    \centering
    \begin{tikzpicture}[align=center,  node distance=2cm, scale=0.5, transform shape]
    \node[main node] (0) {$0$};
    \node[main node] (1) [right of =0]  {$1$};
    \node[main node] (2) [below of =1] {$2$};
    \node[main node] (3) [right of =1] {$3$};
    
    \path[->,draw,thick]
    
    (0) edge node {} (1)
    (1) edge node {} (2)
    (2) edge node {} (3);
\end{tikzpicture}
    \caption{Attack graph under topology B.}
    \label{fig:attackgraphB}
    \end{subfigure}
    \caption{The attack graphs under different network topologies.}
    \label{fig:attackgraphAB}
\end{figure}
Incorporating defender's actions into the attack graph, we obtain the arena of the game, shown in Fig.~\ref{fig:attack-defend-game}. A circle state $(0,A)$ can be understood as the attacker is at node $0$,   the network configuration is $A$, and it is attacker's turn to make a transition. A square state $(1,A)$ can be understood as the attacker is at node $1$,   the network configuration is $A$, and it is defender's turn to make a switch. A transition from circle $(0,A) $ to square $(1,A)$ means that the attacker exploits a vulnerability on node $1$ and reach node $1$. The goal of the attacker is node $3$. That is, if the attacker can reach any of the square states $(3,A)$ or $(3,B)$, then she wins the game. The goal of the defender is to prevent the attacker from reaching the goal. In this game, we can compute the attacker's strategy shown in Fig.~\ref{fig:attack-defend-game} where red, dashed edges indicate the choice of attacker. For example, if the attacker is at host $1$ given topology $B$, she reaches host $2$. If the defender switches to $A$, then she will take action to reach square $(3,A)$. If the defender switches to $B$, then she will take action to reach square $(3,B)$. In this game, there is no winning strategy for the defender given the initial state of the game. In fact, the winning region of the defender is empty.

\begin{figure}
    \centering
        \begin{tikzpicture}[->,>=stealth',shorten >=1pt,auto,node distance=2cm,
                            semithick, scale=0.65, transform shape,  square/.style={regular polygon,  regular polygon sides=4}]
        \tikzstyle{every state}=[fill=white]
        \node[initial,state]   (0A0)                      {$0,A$};
        \node[square,draw]           (1A1) [right  of=0A0]   {$1,A$};
        \node[state]   (1A0)  [above right of= 1A1,xshift=1cm, yshift = 1cm]                    {$1,A$};
        \node[state]   (1B0)   [below right of= 1A1,xshift=1cm, yshift=-1cm]                     {$1,B$};
        \node[square, draw] (2A1) [right of = 1A0, yshift=-2cm] {$2,A$};
          \node[state]   (2A0)   [right of= 2A1]                     {$2,A$};
        \node[square, draw,double ] (3A1) [right  of = 1A0, node distance=4cm] {$3,A$};
        \node[square, draw] (2B1) [right of = 1B0] {$2,B$};
        \node[state] (2B0) [right of = 2B1] {$2,B$};
        \node[square, draw,double] (3B1) [right of = 2B0] {$3,B$};
        \path[->]   
         (0A0) edge[red, dashed]   node    { }           (1A1)
         (1A1) edge[black]   node         { }  (1A0)
         (1A1) edge[black]   node         {}  (1B0)
          (1A0) edge[red, dashed]   node         {}  (2A1)
          (1A0) edge[red, dashed]   node         {}  (3A1)
          (2A1) edge[black]   node         { }  (2A0)
          (1B0) edge[red, dashed] node {} (2B1)
        (2B1) edge[black] node { } (2B0)
                (2B1) edge[black] node { } (2A0)
        (2B0) edge[red, dashed] node {} (3B1)
        (2A1) edge[black] node { } (2B0)
                 (2A0) edge[red, dashed] node {} (3A1);
                \end{tikzpicture}
    \caption{The game transition system given topology switching with simple attacker's reachability objective.}
    \label{fig:attack-defend-game}
\end{figure}
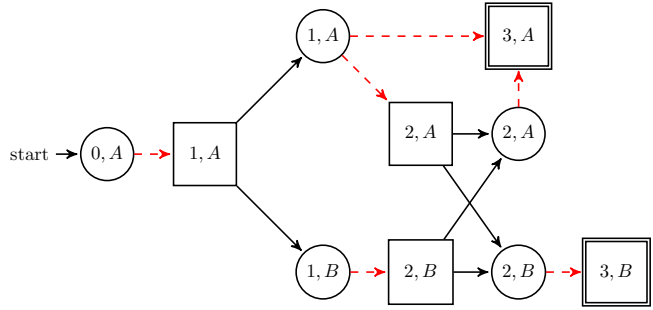

Now, let's consider a different specification of the attacker: $\Eventually (2\land \Eventually 3)$. That is, the attacker must reach node 2 first and then node 3. The \ac{ltl} formula translates to \ac{dfa} in Fig.~\ref{fig:automata}.  Given the new specification, we construct the game transition system in Fig.~\ref{fig:product}. An example of transition $(1,A,q_0,\mbox{circle})\rightarrow (2,A,q_1,\mbox{square})$, where $\mbox{circle}, \mbox{square}$ indicate the shapes of the nodes, is defined jointly by 
$
(1,A,\mbox{circle})\rightarrow(2,A,\mbox{square}),
$ 
and $q_0\xrightarrow{2} q_1$. Given this \ac{ltl} task, the winning strategy of the attacker is indicated with red and dashed edges. It is noted that when the attacker is at node $1$, she will not choose to reach $3$ but to reach $2$, required by the new specification.

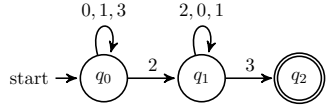
\begin{figure}
    \centering
        \begin{tikzpicture}[->,>=stealth',shorten >=1pt,auto,node distance=2cm,
                            semithick, scale=0.65, transform shape]
        \tikzstyle{every state}=[fill=white]
        \node[initial,state]   (0)                      {$q_0$};
        \node[state]           (1) [ right of=0]   {$q_1$};
        \node[state,accepting]           (2) [right  of=1]   {$q_2$};
         \path[->]
        (0) edge              node        {$2$}       (1)
        (0) edge [loop above] node        {$0,1,3$}       (0)
         (1) edge [loop above] node        {$2,0,1$}       (1)
        (1) edge  node        {$3$}       (2);
        \end{tikzpicture}
    \caption{The automaton representing the attacker's objective}
    \label{fig:automata}
\end{figure}

\begin{figure}
    \centering
        \begin{tikzpicture}[->,>=stealth',shorten >=1pt,auto,node distance=2.5cm,
                            semithick, scale=0.7, transform shape, 
                             square/.style={rectangle}]
        \tikzstyle{every state}=[fill=white]
        \node[initial,ellipse,draw]   (0A0)                      {$0,A,q_0$};
        \node[square,draw]           (1A1) [right  of=0A0]   {$1,A,q_0$};
        \node[ellipse,draw]   (1A0)  [above right of= 1A1]                    {$1,A,q_0$};
        \node[ellipse,draw]   (1B0)   [below right of= 1A1]                     {$1,B,q_0$};
        \node[square, draw] (2A1) [right of = 1A0, yshift=-2cm] {$2,A,q_1$};
          \node[ellipse,draw]   (2A0)   [right of= 2A1]                     {$2,A,q_1$};
        \node[square, draw ] (3A1) [right  of = 1A0, node distance=4.5cm] {$3,A,q_0$};
         \node[square, draw,double ] (3A1q2) [right  of = 2A0] {$3,A,q_2$};
        \node[square, draw] (2B1) [right of = 1B0] {$2,B,q_1$};
        \node[ellipse,draw] (2B0) [right of = 2B1] {$2,B,q_1$};
        \node[square, draw,double] (3B1) [right of = 2B0] {$3,B,q_2$};
        \path[->]   
         (0A0) edge[red, dashed]   node    { }           (1A1)
         (1A1) edge[black]   node         { }  (1A0)
         (1A1) edge[black]   node         {}  (1B0)
          (1A0) edge[red, dashed]   node         {}  (2A1)
          (1A0) edge  node         {}  (3A1)
          (2A1) edge[black]   node         { }  (2A0)
          (1B0) edge[red, dashed] node {} (2B1)
        (2B1) edge[black] node { } (2B0)
                (2B1) edge[black] node { } (2A0)
        (2B0) edge[red, dashed] node {} (3B1)
        (2A1) edge[black] node { } (2B0)
                  (2A0) edge[red, dashed] node {} (3A1q2);
                \end{tikzpicture}
    \caption{The game transition system for \ac{ltl} co-safe formula  $\Eventually (2\land \Eventually 3)$. The red edges are the attacker's strategy.}
    \label{fig:product}
\end{figure}

\end{example}

\subsection{Synthesis of Reactive Defense Strategies with Cyber Deception}

When the attacker has a one-sided misperception of labeling function, as defined in Def.~\ref{def:one-sided-misperception}, the defender \textit{might} strategically utilize this misperception to deceive the attacker into choosing a strategy that is advantageous to the defender. To understand when the defender might have such a deceptive strategy  and how to compute it, we study the solution concept of hypergame.


\paragraph*{Solution Approach} A hypergame $\hgame^2 = \langle \hgame^1, \game_2 \rangle$ is defined using two games, namely $\game_1$ and $\game_2$. Under one-sided misperception of labeling function, defender is aware of both games. Therefore, to synthesize a deceptive strategy, the defender must take into account the strategy that the attacker will use, based on her misperception. That is, the defender must solve two games: Game $\game_2$ to identify the set of states in the game transition system $G_2 \otimes \calA$ that the attacker perceives as winning for her under labeling function $L_2$, and game $\game_1$ to identify the set of states in the game transition system $G_1 \otimes \calA$   that are winning for the defender under (ground-truth) labeling function $L = L_1$. After solving the two games, the defender can integrate the solutions to obtain a set of states, at which the attacker makes mistakes due to the difference between $L_2$ and $L$. Let us introduce a notation to denote these sets of winning states.



\begin{itemize}
    \item $\game_1$: P1's winning region is $\win_1 \subseteq S\times Q$ and P2's winning region is $\win_2 \subseteq S\times Q$.
    \item $\game_2$: P1's winning region is $\win_1^{P2} \subseteq S\times Q$ and P2's winning region is $\win_2^{P2} \subseteq S\times Q$.
\end{itemize}
 
 Figure~\ref{fig:perceptual-games} provides a conceptual representation partitions of the state-space of a game transition system. Due to misperception, the set of states are partitioned into the following regions, 
\begin{itemize}
    \item $\win_1$: is a set of states from which P1 can ensure satisfaction of security objectives, even if P2 has complete and correct information. 
    Thus, P1 can take the winning strategy $\pi_1$.
    \item $\win_1^{P2} \cap \win_2$: is a set of states where P2 is truly winning, but perceives the states to be losing for her; due to misperception.
    Thus, P2 may either give up the attack mission or play randomly.
    \item $\win_2^{P2}\cap \win_2$:  is a set of states in which P2 is truly winning and perceives those states to be winning. 
    In this scenario, she will carry out the winning strategy $\pi_2^{P2}$. However,  this strategy can be different from the true winning strategy $\pi_2$ that P2 should have played if she had complete and correct information. This difference creates unique opportunities for P1 to enforce security of the system.
\end{itemize}

\begin{figure}
    \centering
    \includegraphics[width=0.6\textwidth]{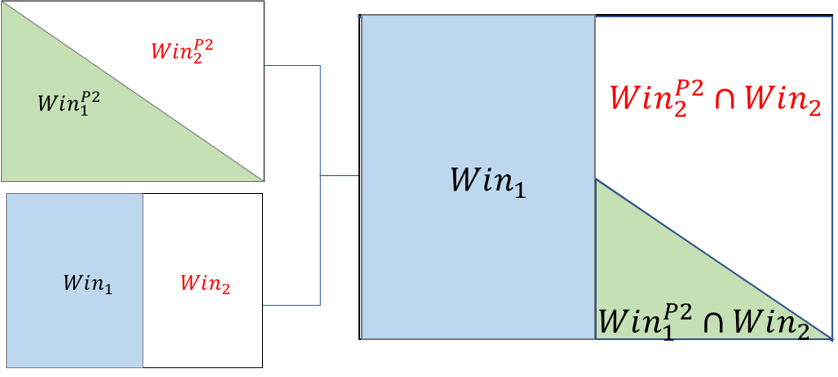}
    \caption{Illustration of the partition given by different perceptual game.}
    \label{fig:perceptual-games}
\end{figure}

To compute deceptive strategy, the defender must reason about (a) how the attacker responds given her perception? and (b) how does the attacker expect the defender to respond, given her perception? 
It is noted that given the winning regions $\win_1^{P2}$ (resp. $\win_2^{P2}$), there exists more than one  strategies that P2 perceives to be winning for P2 (resp. for P1). The problem is to compute a strategy for the defender $\pi_1^\ast$, if exists, such that \emph{no matter which} perceptual winning strategy that P2 selects,  P1 can ensure the security specification is satisfied surely, without contradicting the perception of P2.
  
Given P2's perceptual game $\game_2 =(S\times Q, A, \Delta, (s_0, q_{0,2}), S\times F)$, there can be infinitely many such almost-sure winning mixed
         strategies for P2 \citep{bernetPermissiveStrategiesParity2002}, we take an approximation of the \emph{set} of
         almost-sure winning strategy as a \emph{memoryless set-based
           strategy} as follows.
      \begin{equation}
         \label{eq:approximate-asw}
        \pi^{P2}_2(s,q) = \{a\mid \Delta_2\left((s,q),a\right)\in \win_2^{P2} \}.
      \end{equation}
      In other words, P2 can select any action as long as she can stay within her perceived winning region $\win_2^{P2}$. Given the rational player 2, for a given state $(s,q)$, an action $a $ that is not in $\pi^{P2}_2(s,q)$ is \emph{irrational} as it drives P2 from the perceived sure winning region to the perceived losing region.

      For a state  $(s,q)\in \win_1^{P2}\cap \win_2$, P2 perceives P1 to be winning under labeling function $L_2$, when she is truly winning under ground-truth labeling function $L$. In this case, P1's deceptive strategy should conform to P2's perceptual winning strategy for P1. Otherwise, P2 would know that she is misperceiving the game when she observes P1 deviating from his rational behavior in the perceptual game of P2. 
      When P1's action is inconsistent from what P2 perceives P1 should do, then P2 knows that she have misperception about the game.

      Again,  we take an approximation of the \emph{set} of P1's
         almost-sure winning strategy perceived by P2 as a \emph{memoryless set-based
           strategy} as follows.
      \begin{equation}
         \label{eq:approximate-asw1}
        \pi^{P2}_1((s,q)) = \{a\mid \Delta_2((s,q),a )\in \win_1^{P2} \}.
      \end{equation}

      Next, by removing P2's  actions from $\game_1$ that P2 perceives to be irrational as well as P1's actions that contradicts P2's perception, we obtain a different game.
      Now, we incorporate the knowledge of P1's actions that P2 would perceive to be irrational into the hypergame model. This results in a modified hypergame model as defined below. 
      
     \begin{definition}
     \label{def:hypergame}
     Given the games $\game_1=G_1\otimes \calA  $ constructed using the true labeling function $L$ and  $\game_2=G_2\otimes \calA  $ with P2's misperceived labeling function $L_2$, the \emph{deceptive sure-winning strategy of P1} is the sure-winning strategy of the following game:
     \[
     \hgame = (S\times Q\times Q, A, \bar \Delta, (s_0,q_0, p_0), \win_1 \times Q),
     \]
     where the transition function $\bar \Delta$ is defined such that 
     \begin{itemize}
         \item   For $(s,q,p)\in S_1\times Q \times Q\setminus (\win_1\times Q)$, if $ (s,p)\in \win_1^{P2}$, then actions in $\pi^{P2}_1(s,p) $ 
        are enabled. Otherwise, all actions $a\in A_1$ are enabled. For each enabled action $a$, let 
         $\bar \Delta((s,q,p),a) =(s',q',p')$ where $s' = T(s, a)$, $q'= \delta(q, L(s'))$ and $p'= \delta(p, L_2(s'))$. 
         \item      For $(s,q,p)\in S_2\times Q \times Q\setminus (\win_1\times Q)$, if $(s,p)\in \win_2^{P2} $, then actions in $\pi^{P2}_2(s,p) $ 
        are enabled. Otherwise, all actions $a\in A_2$ are enabled. For each enabled action $a$, let $\bar \Delta((s,q,p),a) =(s',q',p')$ where $s' = T(s, a)$, $q'= \delta(q, L(s'))$ and $p'= \delta(p, L_2(s'))$.
        \item The initial state is $(s_0,q_0,p_0)$ where $(s,q_0)$ is the initial state in $\game_1$ and $(s_0, p_0) $ is the initial state in $\game_2$.
     \end{itemize} 
     \end{definition}

The transition function can be understood as follows. At a P1 state, when P2 perceives a state $(s, q, p)$ to be winning for P1, the permissive actions in $\pi_1^{P2}$ of P1  are enabled at that state. Otherwise, P2 would assume that P1 may choose any action from $A_1$. Similarly, at a P2 state, when P2 perceives a state to be winning for herself, she might choose any action from her permissive action set $\pi_2^{P2}$. Whereas, if P2 perceives the current state $(s, q, p)$ to be losing for her, given her perception, she would choose any action from $A_2$.

\begin{lemma}
The sure-winning strategy $\pi_1^\ast$ of the game $\hgame$ in Def.~\ref{def:hypergame} is stealthy  for any state $(s,q,p)$ where $(s,q)\in \win_2$ as it   does not reveal any information with which P2 can deduce that some misperception exists.
\end{lemma}
\begin{proof}
 In P2's perceived winning region $\win_2^{P2}$ for herself, any strategy of P1 is losing. Thus $\pi_1^\ast$ will not contradict P2's perception. In P2's perceived winning region $\win_1^{P2}$ for P1, an action, which is selected by $\pi_1^\ast$ (if defined for that state), ensures that P1 to stay within $\win_1^{P2}$ and thus will not contradict the perception of P2. In both cases, P2 will not deduce the fact that there is a misperception.
\end{proof}
\begin{example}[Continued]
Let us continue with the toy example and the simple reachability objective $\Eventually 3$ (eventually reach node 3). Suppose the node $2$ is a decoy. Then the attacker's labeling function $L_2$ differs from the defender's labeling function $L$: $L_2((2,X))= \emptyset$ and $L((2,X)) = \mbox{decoy}$ where $X \in \{A,B\}$. As the attacker is to avoid reaching decoys, if she knows the true labeling $L$, then at the attacker's circle state $(1,A)$ in Fig.~\ref{fig:attack-defend-game-v2}, she will not choose to reach node $2$ $($and thus the square state $(2,A))$. The attacker's strategy given the true labeling function is given in Fig.~\ref{fig:attack-defend-game-v2}. In this game, the attacker has no winning strategy to reach $3$ from $0$ in the true game as the defender can choose to switch to topology $B$ $((1,A,\mbox{square})\rightarrow (1,B,\mbox{circle}))$. However, with the misperception on the labeling function, the attacker believes that she has a winning strategy from node $0$ (see the attacker's strategy in Fig.~\ref{fig:attack-defend-game}).
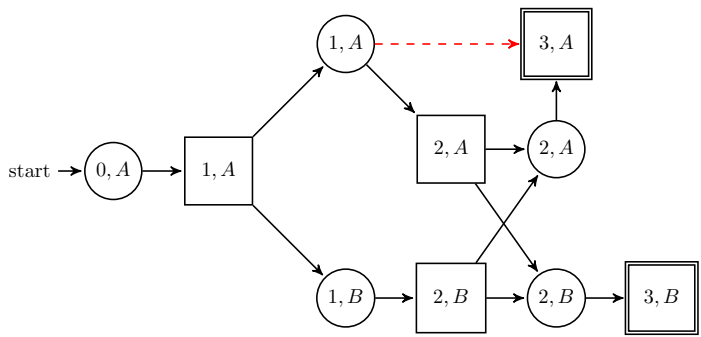
\begin{figure}[ht]
    \centering
        \begin{tikzpicture}[->,>=stealth',shorten >=1pt,auto,node distance=2cm,
                            semithick, scale=0.7, transform shape,  square/.style={regular polygon,  regular polygon sides=4}]
        \tikzstyle{every state}=[fill=white]
        \node[initial,state]   (0A0)                      {$0,A$};
        \node[square,draw]           (1A1) [right  of=0A0]   {$1,A$};
        \node[state]   (1A0)  [above right of= 1A1,xshift=1cm, yshift = 1cm]                    {$1,A$};
        \node[state]   (1B0)   [below right of= 1A1,xshift=1cm, yshift=-1cm]                     {$1,B$};
        \node[square, draw] (2A1) [right of = 1A0, yshift=-2cm] {$2,A$};
          \node[state]   (2A0)   [right of= 2A1]                     {$2,A$};
        \node[square, draw,double ] (3A1) [right  of = 1A0, node distance=4cm] {$3,A$};
        \node[square, draw] (2B1) [right of = 1B0] {$2,B$};
        \node[state] (2B0) [right of = 2B1] {$2,B$};
        \node[square, draw,double] (3B1) [right of = 2B0] {$3,B$};
        \path[->]   
         (0A0) edge  node    { }           (1A1)
         (1A1) edge[black]   node         { }  (1A0)
         (1A1) edge[black]   node         {}  (1B0)
          (1A0) edge   node         {}  (2A1)
          (1A0) edge[red, dashed]   node         {}  (3A1)
          (2A1) edge[black]   node         { }  (2A0)
          (1B0) edge  node {} (2B1)
        (2B1) edge[black] node { } (2B0)
                (2B1) edge[black] node { } (2A0)
        (2B0) edge  node {} (3B1)
        (2A1) edge[black] node { } (2B0)
                 (2A0) edge  node {} (3A1);
                \end{tikzpicture}
    \caption{The game transition system and the attacker's winning strategy given that she is to reach node $3$ and knows that node $2$ is decoy.}
    \label{fig:attack-defend-game-v2}
\end{figure}

Consider now the initial state is $(1,A,\mbox{circle})$, that is, the attacker has compromised node $1$ and the current network topology is $A$. If the attacker knows $2$ is a decoy, then she will choose to reach $3$ deterministically. If the attacker does not know $2$ is a decoy, then she is indifferent to reaching node $2$ or node $3$, because in her  perception, these two actions ensures that with probability one, she can reach node $3$ in finitely many steps. Given this analysis, it is not difficult to see that $(1,A, \mbox{circle})$ is sure-winning for the attacker given the true game, but positive winning for the defender given the perceptual game of the attacker, with the incorrect labeling. Here, positive winning means that the defender wins with a positive probability--that is, the probability when the attacker makes mistakes (visiting decoy node $2$) due to her misperception.

Finally, we construct the hypergame $\hgame$ in Fig.~\ref{fig:hypergame-ex} where the defender's objective is $\neg p \until \mbox{decoy}$ where $p$ is an atomic proposition that evaluates true when node $3$ is compromised. The labeling functions are: For $X\in \{A,B\}$,  $L (1,X)= L_2(1,X)=\emptyset$, $L (3,X)= L_2(3,X)=\{p\}$ and $L (2,X)= \mbox{decoy}$ but $L_2(2,X)=\emptyset$. The states (shaded, red) are when the attacker's perceived automaton  state  differs from the defender's automaton state.  For example, $((2,A),q_2,q_0,\mbox{circle})$ means that the defender knows that the attacker reached a decoy but the attacker is unaware of this fact.
\begin{figure}
\begin{subfigure}[b]{0.2\textwidth}
    \centering
        \begin{tikzpicture}[->,>=stealth',shorten >=1pt,auto,node distance=2cm,
                            semithick, scale=0.5, transform shape]
        \tikzstyle{every state}=[fill=white]
        \node[initial,state]   (0)                      {$q_0$};
        \node[state]           (1) [ right of=0]   {$q_1$};
        \node[state,accepting]           (2) [ below right of=0]   {$q_2$};
         \path[->]
        (0) edge              node        {$p$}       (1)
        (0) edge [loop above] node        {$\emptyset$}       (0)        
        (0) edge [right]   node        {$\mbox{decoy}$}       (2)
        (2) edge [loop right] node {$\top$} (2)
        (1) edge [loop above] node {$\top$} (1);
        \end{tikzpicture}
    \caption{}
    \label{fig:dfa_simple}
    \end{subfigure}
    \begin{subfigure}[b]{0.79\textwidth}
        \begin{tikzpicture}[->,>=stealth',shorten >=1pt,auto,node distance=3.5cm,
                            semithick, scale=0.5, transform shape,  square/.style={rectangle}]
        \tikzstyle{every state}=[fill=white,draw,ellipse]
        \node[initial,draw,ellipse]   (0A0)                      {$(0,A),q_0,q_0$};
        \node[square,draw]           (1A1) [right  of=0A0]   {$(1,A),q_0,q_0$};
        \node[draw,ellipse]   (1A0)  [above right of= 1A1 ]                    {$(1,A),q_0,q_0$};
        \node[draw,ellipse]   (1B0)   [below right of= 1A1 ]                     {$(1,B),q_0,q_0$};
        \node[square, draw,  fill=red!10] (2A1) [below right of = 1A0,xshift=1cm] {$(2,A), q_2, q_0$};
          \node[draw,ellipse,  fill=red!10]   (2A0)   [right of= 2A1]                     {$(2,A),q_2,q_0$};
                  \node[square, draw , fill=red!10] (3A1q2) [right  of = 2A0] {$(3,A),q_2,q_1$};

        \node[square, draw ] (3A1q1) [  right  of = 1A0] {$(3,A),q_1,q_1$};
        \node[square, draw, fill=red!10] (2B1) [right of = 1B0] {$(2,B), q_2,q_0$};
        \node[draw,ellipse, fill=red!10] (2B0) [right of = 2B1] {$
        (2,B),q_2,q_0$};
        \node[square, draw,  fill=red!10] (3B1) [right of = 2B0] {$(3,B),q_2,q_0$};
        \path[->]   
         (0A0) edge [red, dashed]   node    { }           (1A1)
         (1A1) edge[black]   node         { }  (1A0)
         (1A1) edge[black]   node         {}  (1B0)
          (1A0) edge [red,dashed]  node         {}  (2A1)
          (1A0) edge [red, dashed]  node         {}  (3A1q1)
          (2A1) edge[black]   node         { }  (2A0)
          (1B0) edge[red, dashed] node {} (2B1)
        (2B1) edge[black] node { } (2B0)
                (2B1) edge[black] node { } (2A0)
        (2B0) edge[red, dashed] node {} (3B1)
        (2A1) edge[black] node { } (2B0)
                 (2A0) edge[red, dashed] node {} (3A1q2);
                \end{tikzpicture}
    \caption{}
    \end{subfigure}
    \caption{(a)The \ac{dfa} for the defender's objective $\neg  p\until \mbox{decoy}$. (b) The game that P1 uses to compute deceptive sure-winning strategy. The red and dashed edges are perceived winning actions of the attacker.}
    \label{fig:hypergame-ex}
\end{figure}
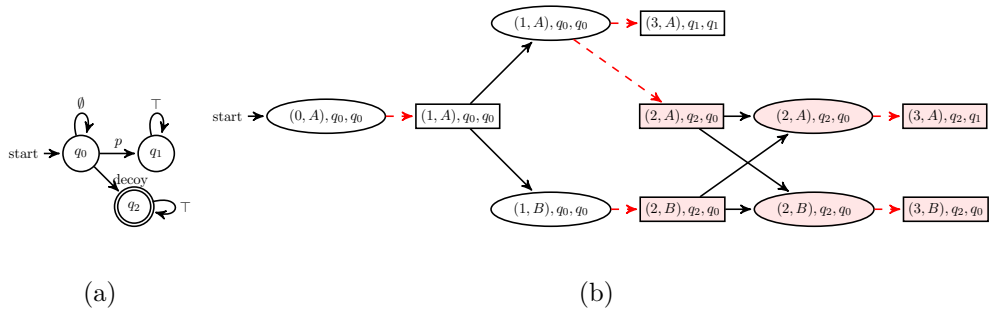
\end{example}

\section{Case Study}

We consider a simple network system illustrated in Fig.~\ref{fig:network-ex}. 
\begin{figure}[ht]
    \centering
    \includegraphics[width=0.5\textwidth]{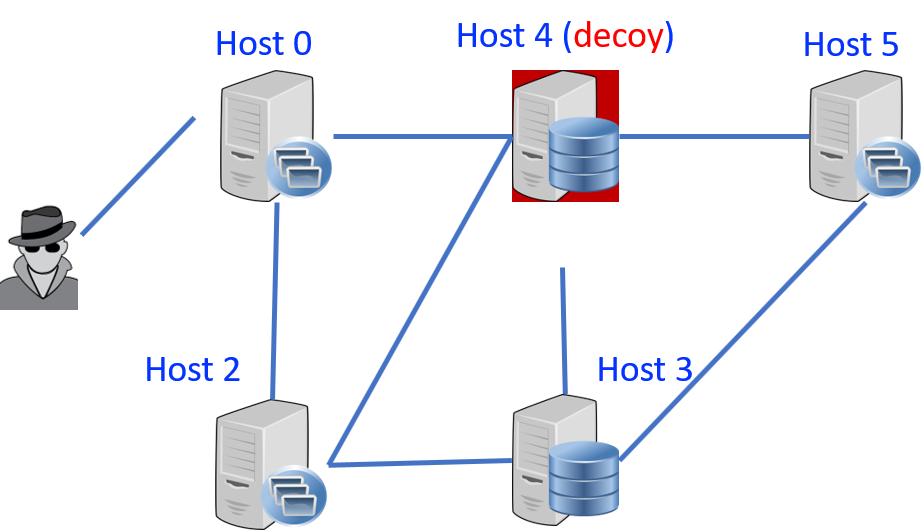}
    \caption{A example of network system.}
    \label{fig:network-ex}
\end{figure}

\begin{table}
\centering
\caption{The pre- and post-conditions of vulnerabilities.\label{tb:vuldefine}}{%
\begin{tabular}{||c|l||}
\toprule
Vulnerability ID &  Pre- and Post- Conditions\\
\midrule
0 & Pre : $c\ge 1$, service $0$ running on target host, \\
  &Post : $c=2$, stop service $0$ on the target, reach target host.\\
\hline
1 & Pre: $c\ge 1$, service $1$ running on the target host, \\
  & Post : reach target host. \\
\hline
2 & Pre: $c\ge 1$, service $2$ running on the target host \\
     & Post : $c=2$,  reach target host.\\
\botrule
\end{tabular}}{}
\end{table}

\begin{table}
\caption{The defender's options.
\label{tb:networkstatus_simple}}{%
\begin{tabular}{@{}ccc@{}}
\toprule
Host ID & Services & Non-critical Services  \\
\midrule
0 & $\{0,1,2\} $ & $\emptyset$ \\
 1& $\{ 1,0\}$ &$\emptyset$\\
 2& $\{1,2\}$ & $\emptyset$\\
 3 & $\{0,1,2\}$ &$\{0,1\}$\\
 4 & $\{0,1\}$ &$\emptyset$\\
 5 & $\{0,1,2\}$ &$\emptyset$\\
\botrule
\end{tabular}}{}
\end{table}



In this network, each host runs a subset $\servs = \{0,1,2\}$ of services. A user in the network can have one of the three login credentials $\mbox{credentials}= \{0,1,2\}$ standing for ``no access'' (0), ``user'' (1), and ``root'' (2). There are a set of vulnerabilities in the network, each of which is defined by a pre-condition and a post-condition. The pre-condition is a Boolean formula that specifies the set of logical properties to be satisfied for an attacker to exploit the vulnerability instance. The post-condition is a Boolean formula that specifies the logical properties that can be achieved after the attacker has exploited that vulnerability.  The set of vulnerabilities are given in Table~\ref{tb:vuldefine} and  are generated based on the vulnerabilities described in \citep{Jha2002Two}.

The defender can temporally suspend noncritical services from servers. To incorporate this defense mechanism, we assign each host a set of noncritical services that can be suspended from the host.  In Table~\ref{tb:networkstatus_simple}, we list the set of services running on each host, and a set of noncritical services that can be suspended by the defender. Other defenses can also be considered. For example, if the network topology can be reconfigured online, then the state in the game arena should keep track of the current topology configuration of the network. In our experiment, we consider simple defense actions. However, our method extends to more complex defense mechanisms.

The attacker, at a given
attacker's state, can exploit any existing vulnerability on the
current host.  The defender, at a defender's state, can choose to suspend
a noncritical service on any host in the network. 
The attacker's objective is expressed using \ac{scltl} formula 
\[
\varphi_2 = \neg\mbox{decoy}\until p_2 \land \neg \mbox{decoy} \until  p_5,
\]
where $p_i$ means that the attacker has compromised host $i$ and gained user or root access on that machine.
However, the attacker does not know the location of the decoys. In this system, decoy is host $4$. 

The following result is obtained from hypergame analysis: at the initial state, the attacker is at host $0$ with user access.
\begin{itemize}
\item The initial state is perceived to be winning by the attacker. 
\item Assuming complete, symmetric information, the size of winning region for the defender is 131.
    \item With asymmetric information, when the attacker plays a perceived winning strategy, the size of winning region for the defender is 193--which is greater than that with symmetric information.
    \item The defender has a winning strategy to prevent the attacker from achieving her objective in the network, using the solution of the hypergame.
    \end{itemize}
    
It is noted that, using the winning region of the defender, we know for a given initial state, whether the security specification is satisfied. For a different initial state, for example, the attacker visited host 2 with user privileges, we can directly examine whether the security is ensured by checking if the new initial state is in the winning region for the defender. This set provide us important insight to understand the weak points in the network, and can be used for guide the decoy allocation.

\section{Conclusion}
The goal of formal synthesis is to design dynamic  defense with guarantees on critical security specifications in a cyber network. 
In this chapter, we introduced a game on graph model for capturing the attack-defend interactions in a cyber network for reactive defense, subject to security specifications in temporal logic formulas. In reactive defense, the defender can take actions in response to the exploit actions of the attacker. We introduced a hypergame for games on graphs to capture payoff misperception for the attacker caused by the decoy systems. The solution concept of hypergames enables us to synthesize effective defense strategy given the attacker's misperception of the game, without contradicting the belief of the attacker. There are multiple extensions from this study: The framework assumes asymmetric information but complete observations for both defender and attacker. It is possible to extend for defense design with a partially observable defender/attacker. By examining the winning region, it provides important insights for the resource allocation of decoy systems.

\backmatter

\bibliographystyle{plainnat}
\bibliography{refs}%



\printindex

\end{document}